\documentclass[a4paper,12pt]{article}
\usepackage{latexsym,amssymb,amsfonts,amsmath,amsthm}
\usepackage[dvips]{graphicx}
\usepackage{color}

\newtheorem{theorem}{Theorem}
\newtheorem{lemma}{Lemma}
\newtheorem{corollary}{Corollary}

\newtheorem{assumption}{Assumption}

\numberwithin{theorem}{section}
\numberwithin{lemma}{section}
\numberwithin{corollary}{section}
\numberwithin{proposition}{section}
\numberwithin{remark}{section}

\newcommand{\bs}[1]{\boldsymbol{#1}}

\newcommand{\spann}{{\rm span}}

\newcommand{\supp}{{\rm supp}}

\newcommand{\dist}{{\rm dist}}

\newcommand{\red}{\textcolor{black}}
\newcommand{\blue}{\textcolor{black}}

\title{{\Large {\bf Dynamical system induced by quantum walk }}
\author{
{\small Yusuke Higuchi\footnote{Email: higuchi@cas.showa-u.ac.jp }}\\
{\scriptsize Mathematics Laboratories, College of Arts and Sciences, Showa University, }\\
{\scriptsize Fujiyoshida, Yamanashi 403-0005, Japan.}\\
{\small Etsuo Segawa\footnote{Email: segawa-etsuo-tb@ynu.ac.jp}}\\
{\scriptsize Graduate School of Education Center, Yokohama National University} \\
{\scriptsize Hodogaya, Yokohama 240-8501, Japan.}}
}
\date{\empty }
\begin{document}
\maketitle

\par\noindent
\begin{small}
\par\noindent
{\bf Abstract}. 
We consider the Grover walk model on a connected finite graph with two infinite length tails and 
we set an $\ell^\infty$-infinite external source from one of the tails as the initial state. 
We show that for any connected internal graph, a stationary state exists, moreover 
a perfect transmission to the opposite tail always occurs in the long time limit. 
We also show that 
the lower bound of the norm of the stationary measure restricted to the internal graph is proportion to the number of edges of this graph. 
Furthermore when we add more tails (e.g., $r$-tails) to the internal graph, then we find that from the temporal and spatial global view point, 
the scattering to each tail in the long time limit coincides with the local one-step scattering manner of the Grover walk at a vertex whose degree is $(r+1)$. 
\end{small}

\section{Introduction}
It is well known that 
for the discrete-time isotropic random walk on a connected finite graph, the probability distribution 
converges to the stationary distribution which is proportion to the $(+1)$ (; that is, the maximal) -eigenvalue's eigenvector of the transition matrix. 
On the other hand, for a quantum walk whose time evolution is a unitary operator, 
its stationarity with a natural initial state; e.g., starting from an arbitrary arc, cannot be expected to be described in such a way in general 
because the spectrum of the time evolution operator is distributed on the unit circle in the complex plane, 
which means every eigenvalue $\lambda$ satisfies $|\lambda|=1$. 
\blue{
There exist many kinds of comparison studies between classical and quantum settings. 
For example, the hitting times on some family of graphs by classical algorithms based on a random walk and quantum algorithms based on 
a so-called continuous-time quantum walk are discussed in~\cite{FaGoGu,FaGu}. 
Of course, we can find other works on classical and quantum algorithms and hitting time in the discrete-time, 
for instance, in~\cite{Portugalbook} and its references therein.
}

Turning our eyes towards ``stationary state" of infinite graphs, especially for the (infinite) one-dimensional lattice $\mathbb{Z}$, 
we can see a series of works finding a stationary measure $\mu_*$ of quantum walks in the following meaning: 
let $\psi_n$ taking complex value at each arc of the one-dimensional lattice be the $n$-th iteration of a quantum walk on one-dimensional lattice
whose initial state $\psi_0$ no longer needs to square summable.  
Then we call $\mu_*: \mathbb{Z}\to [0,\infty)$ a stationary measure if and only if 
$\sum_{t(a)=j}|\psi_n(a)|^2=\mu_*(j)$ for any natural number $n$. 
\red{Note that a stationary state is derived from a generalized eigenfunction immediately, but this is not everything  
as is suggested by an interesting example in~\cite{KonnoTakei}. 
}
Thus finding such stationary measure $\mu_*$ is a non-trivial problem and related to 
\red{trying to understand the behaviour of quantum walks as a kind of {\it process\/}}. 
\blue{Towards addressing this challenging problem, a classification of the generalized eigenfunctions' shape of quantum walks on $\mathbb{Z}$ 
are gradually revealed with respect to the absolute values, e.g., uniformity~\cite{Konno}, the support 
finiteness~\cite{KomatsuKonno}, polynomially increasing~\cite{KonnoTakei} and exponentially decreasing~\cite{KonnoLuczakSegawa}. 
Around them, the transfer matrix, which is a standard method in the spectral analysis on the CMV matrix, e.g., \cite{GZ},
is rediscovered in the context of quantum walks~\cite{KawaiKomatsuKonno} and applied to obtain some generalized eigenfunctions efficiently. 
These results depending on just the some ``static" class of eigenfunctions are quite interesting, but does not 
reflect the behaviour of the unitary dynamics which the quantum walk originally reveals.  
Thus it is natural to ask for a meaning of these stationary measures and states from the view point of dynamics which 
quantum walks originally have; some results can be seen in \cite{FelHil1,FelHil2,MMOS}.}

\blue{Let us give a short review on~\cite{MMOS}, where a simple quantum walk
 model on $\mathbb{Z}$ is introduced and its dynamics is discussed. In result, 
such a quantum walk is succeeded in connecting  
the quantum dynamics of double well 
potentials provider's a quantum resonance effect.}
In the setting of this  model, all quantum coins assigned to all the vertices except the two vertices $\{0,m\}$ produce a free-walk; that is, no reflection,  
while the quantum coins assigned at the exceptional two vertices produce a non-trivial scattering, 
and the initial state is externally injected at each time step from the negative side on the one-dimensional lattice. 
Then it is obtained that there is a stationary state as the limit of the time iterations from this initial state 
and this stationary state coincides with one of the bounded generalized eigenfunction of the 
whole system of unitary time evolution operator obtained in \cite{KawaiKomatsuKonno}. 
Moreover it is shown that this model includes the quantum walk model whose stationary state coincides 
with the solution of the stationary Schr{\"o}dinger equation on the metric one-dimensional lattice with the double well delta potential~\cite{Albe}. 
\red{Here the notion of a `quantum graph walk' introduced by~\cite{HKSS} plays an important role.} 

\blue{In this paper, we tend to extend the model from that in~\cite{MMOS}: 
(1) we generalize the connected finite graph, the internal graph; 
(2) we increase the number of tails, that is, the number of directions for 
observing the behaviour of scattering on the internal; 
(3) we observe the distribution of penetration into the internal, that is, 
a kind of conditional probability on the internal. 
In the next section, 
using a simple example, 
the internal graph is a $3$-cycle with two tails, 
we demonstrate and illustrate what we intend. 
We first expect that the stationary state of a quantum walk can be 
obtained by a balance between the inflow and the outflow in two-tail model.
In result, we obtain the existence of stationarity \cite{FelHil1,FelHil2} 
by a notion of dynamical system~\cite{R}. Moreover, 
using some by-product of this proof 
and the eigenspace of Grover walk induced by homological structure of graph,  
we characterized some detailed properties of the stationary state 
and the penetration state. }


\blue{To explain our setting and results more precisely, 
let us first give some notations and definitions of 
Grover walk on a graph.}
Let $G=(V,E)$ be a connected graph. 
Each edge $e\in E$ produces the two kinds of arcs $a$ and $\bar{a}$, where $\bar{a}$ is called the inverse arc of $a$. 
The set of arcs is denoted by $A$. Remark that $a\in A$ if and only if $\bar{a}\in A$\red{; we sometimes} call $A$ a symmetric arc set. 
The total Hilbert space of the Grover walk is generated by $A$; that is, the set of all the functions taking a complex value at each arc. 
The time evolution is \red{given} by the following unitary operator $U$: 
	\begin{equation}\label{eq:Grover_time_evolution} 
        (U\psi)(a)= -\psi(\bar{a})+\sum_{b:t(b)=o(a)}\frac{2}{\mathrm{deg}(o(a))}\psi(b). 
        \end{equation}
Here 
\red{the origin and terminal vertices of $a\in A$ are denoted by $o(a)$ and $t(a)$, respectively. 
Thus it holds that $o(\bar{a})=t(a)$, $t(\bar{a})=o(a)$. 
For example, if $u$ and $v$ is connected by the single edge, the arc $a$ from $u$ and to $v$ is denoted by $(u,v)$. 
Then $o(a)=u$ and $t(a)=v$ hold. 
Moreover} $\mathrm{deg}(u)$ is the degree of the vertex $u$; that is, $\mathrm{deg}(u)=|\{a\in A \;|\; t(a)=u\}|=|\{a\in A \;|\; o(a)=u\}|$. 
Let $\psi_n$ be the $n$-th iteration of the Grover walk, that is, $\psi_{n+1}=U\psi_n$ with some initial state. 
For arbitrary $u\in V$, if we put $a_j\in A$ such that $t(a_j)=u$ for $j=1,\dots,\mathrm{deg}(u)$, 
then the one-step local Grover walk's scattering at the vertex $u$ is described by 
	\[ \bs{\omega}_{out}^{(n+1)}(u)=\mathrm{Gr}(\mathrm{deg}(u)) \bs{\omega}_{in}^{(n)}(u). \]
Here for $m\in \mathbb{N}$, 
	\begin{align*} 
        \bs{\omega}_{in}^{(m)}(u) &= {}^T[\psi_m(a_1),\dots,\psi_m(a_{\mathrm{deg}(u)})]; \\
        \bs{\omega}_{out}^{(m)}(u)&={}^T[\psi_m(\bar{a}_1),\dots,\psi_m(\bar{a}_{\mathrm{deg}(u)})], 
        \end{align*}
and $\mathrm{Gr}(d)$ is the $d$-dimensional Grover matrix, that is, $(\mathrm{Gr}(d))_{ij}=2/d-\delta_{ij}$.

Let $G_0=(V_0,E_0)$ be the internal finite graph and the induced symmetric arc set be denoted by $A_0$. 
The degree of $u\in V_0$ in $G_0$, $\mathrm{deg}_{G_0}(u)$, is denoted by $d(u)$. 
We arbitrarily choose two vertices $u_{+}$ and $u_{-}$ from the vertex set $V_0$. 
We join two additional semi-infinite length paths to the input and output vertices $u_{+}$ and $u_{-}$, respectively. 
We denote the input and output tails  joined to $u_{+}$ by $\mathbb{P}_+$ and joined to $u_{-}$ by $\mathbb{P}_-$, respectively. 
The graph adding the two infinite length tails to $G_0$ is denoted by $\tilde{G}=(\tilde{V},\tilde{E})$ with the symmetric arc set $\tilde{A}$ 
and the degree of $u\in \tilde{V}$ in $\tilde{G}$, \red{$\mathrm{deg}_{\tilde{G}}(u)$,} is denoted by $\tilde{d}(u)$. 

We inject the $+1$ external source to the input vertex $\red{u_{+}}$ at every time step, in other words, 
we set the initial state by 
	\[
        \psi_{in}(a)=
        \begin{cases}
        1 & \text{: $t(a)\in V(\mathbb{P}_+)$,\;$\mathrm{dist}(t(a),u_+)<\mathrm{dist}(o(a),u_+)$} \\
        0 & \text{: otherwise.}
        \end{cases}
        \] 
Iterating the Grover time evolutions on $\tilde{G}$, we can state that
a new quantum walker continuously comes from $u_+$ at every time step, while 
once a quantum walker goes outside of the \red{internal} graph $G_0$, then 
she never goes back to \red{$G_0$} since the Grover walk dynamics on the two-tail is \red{the free quantum walk} by (\ref{eq:Grover_time_evolution}). 
Under such a situation, we \red{shall} take the Grover walk's time iterations many times. 
Now the following natural questions may arise: 
\begin{enumerate}
\item Does the stationary state exist ? 
\item If the stationary state $\psi_\infty$ exists, what is the shape of $\psi_\infty$, especially, 
the transmission rate $t_*$ and the reflection rate $r_*$ as the outflow from $u_-$ and $u_+$, respectively in the long time limit; 
that is, for $b\in A(\mathbb{P}_-)$ with $o(b)=u_-$, and $a\in A(\mathbb{P}_+)$ with $o(a)=u_+$, what are the following values
	\[t_*:=|\psi_{\infty}(b)|^2,\;r_*:=|\psi_{\infty}(a)|^2 ? \]
and also what kinds of graph structure \red{are reflected the shape of $\psi_\infty$} ? 
\end{enumerate}
\red{The answer for (1) has been obtained by \cite{FelHil1,FelHil2} (2005, 2007) as follows. }
\begin{theorem}(\label{thm:existence}\cite{FelHil2})
For any connected finite graphs $G_0=(V_0,E_0)$ and for any connected manner of two infinite tails, 
the stationary state uniquely exists; that is, 
	\[ \lim_{n\to\infty}U^n\psi_{in}={}^\exists \psi_\infty. \]
\end{theorem}
Remark that $\psi_{in}$ and $\psi_{\infty}$ belongs to $\ell^\infty$ category.
The existence of the stationary state under more general setting, \red{Assumption 1 (cf.~\cite{FelHil2}), can be seen in Theorem~\ref{thm:existance2} in Section~3. }
Although there may be similar statements to \cite{FelHil2}, 
we emphasize that we give an understanding of the result on \cite{FelHil2} from the view point of the {\it Jordan decomposition} of 
a dynamical mapping~\cite{R} since this map loses the normality due to the cut-off with respect to the internal graph 
of the unitary operator on the whole system. 

By using this existence of the stationary state, we can further proceed the analysis of the stationary state 
to find specialities of the Grover walk in the following theorems. 
These are our main results. 
\begin{theorem}\label{thm:perfect_transmittion}
Let us consider the Grover walk model on a finite \red{internal} graph with two tails and keep inserting inflow from one tail with the amplitude $1$. 
Then for any connected graphs $G_0=(V_0,E_0)$ and for any connected manner of two infinite tails,
the perfect transmission always happens; that is, $t_*=1$ and $r_*=0$. 
Moreover the stationary state $\psi_\infty$ satisfies the following properties 
\begin{align}
\sum_{a\in A(\tilde{G}):t(a)=u} \psi_\infty(a) &=\sum_{a\in A(\tilde{G}):o(a)=u} \psi_\infty(a)=\tilde{d}(u)/2\;\;\;\;({}^\forall u\in V(\tilde{G})); \label{eq:stationary_vertex}\\
\psi_\infty(a)+\psi_\infty(\bar{a}) &=1\;\;\;\;({}^\forall a\in A(\tilde{G})). \label{eq:stationary_edge}
\end{align}
\end{theorem}
The first equation (\ref{eq:stationary_vertex}) implies that the average with respect to all the arcs whose terminus (origin) are $u$, is $1/2$ 
while the second equation (\ref{eq:stationary_edge}) implies that the average with respect to two arcs, 
whose support edges are the same, is also $1/2$. 
Although the notion of \red{so-called} a perfect state transfer of discrete-time quantum walks e.g.,~\cite{Stefanak} is different from our perfect transmission $t_*=1$,
finding its connections is one of the interesting future's problems.  
\red{Moreover let us see a relation to \cite{FaGu} 
considering a \red{so-called} continuous-time quantum walk on binary trees as the internal graph $G_0$: 
if the energy of the incident wave is $E=0$, then the perfect transmitting occurs in this model. 
The Grover walk, which we treat in this paper, corresponds to the potential free case of the quantum graph~\cite{HKSS} 
which is the stationary Schr{\"o}dinger equation on a metric graph. 
Since the constant value $1$ is inputted into the internal graph at every step in our case, the energy corresponds to $0$. 
Therefore the perfect transmitting of the Grover walk with the tail number $2$ seems to be consistent with \cite{FaGu}. }

By applying the Cauchy-Schwartz inequality to (\ref{eq:stationary_vertex}) or (\ref{eq:stationary_edge}), we obtain the following corollary 
\red{which implies a penetration into the internal graph always occurs in the large time behavior.} 
\begin{corollary}
The total mass in the internal graph $G_0$ has the following lower boundary: 
\[ \sum_{a\in A(G_0):t(a)\in V_0} |\psi_\infty(a)|^2 \geq  \frac{|E_0|}{2}. \]
\end{corollary}
We also consider a natural extension of the number of tails from $2$ to $r\geq 2$. 
Then we obtain the following interesting result which shows that a temporal and spatial global view point of 
our Grover walk model can be interpreted as the one-step local Grover walk's scattering on a vertex of degree $r$. 
More precisely, we obtain the following theorem. 
\begin{theorem}\label{thm:scattering}
Consider an infinite graph $\tilde{G}$ constructed of a finite \red{internal} graph $G_0=(V_0,A_0)$ and $r$-tails $\mathbb{P}_1,\dots,\mathbb{P}_r:$ 
	\[ \tilde{G}=G_0 \cup \bigcup_{j=1}^{r}\mathbb{P}_j \mathrm{\;with\;} V(\mathbb{P}_j)\cap V_0 =\{u_j\} \mathrm{\;for\;} j=1,\dots,r. \]
Let $e_j\in A(\mathbb{P}_j)$ such that $t(e_j)=u_j$ for $j=1,\dots,r$. 
Assume that the amplitude of the inflow on $e_j$ is $\alpha_j\in \mathbb{C}$. 
Then there exists a stationary state $\psi_\infty$ and we have 
	\begin{align}
	\sum_{a\in A(\tilde{G}): t(a)=u}\psi_\infty(a) &= \sum_{a\in A(\tilde{G}): o(a)=u}\psi_\infty(a)
        	=\mathrm{ave}(\alpha_1,\dots,\alpha_r)\; \tilde{d}(u) \;\; (\forall u\in \tilde{V}); \label{eq:shinbashi}\\
        \psi_\infty(a)+\psi_\infty(\bar{a}) 
        	&= 2\mathrm{ave}(\alpha_1,\dots,\alpha_r) \;\; (\forall a\in \tilde{A}), \label{eq:shinbashi2}
	\end{align}
where $\mathrm{ave}(\alpha_1,\dots,\alpha_r)$ is the average of $\alpha_1,\dots,\alpha_r$. 
Moreover 
	\begin{equation}\label{eq:GroverScat} \Omega_{out}(G_0)=\mathrm{Gr}(r)\Omega_{in}(G_0) \end{equation}
holds, where 
	\begin{align*} 
        \Omega_{in}(G_0) &= {}^T[\alpha_1,\dots,\alpha_r]={}^T[\psi_\infty(e_1),\dots,\psi_\infty(e_r)]; \\
        \Omega_{out}(G_0) &= {}^T[\psi_\infty(\bar{e}_1),\dots,\psi_\infty(\bar{e}_r)]. 
        \end{align*} 
\end{theorem}
\noindent Since $\mathrm{Gr}(r)$ is self-adjoint unitary, Theorem~\ref{thm:scattering} implies that 
for arbitrary output flow $\Omega_{out}(G_0)$, there exists an input flow $\Omega_{in}(G_0)$ to accomplish the output flow $\Omega_{out}(G_0)$ 
using this Grover walk model; 
the input flow is of the form $\Omega_{in}(G_0)=\mathrm{Gr}(r) \Omega_{out}(G_0)$. 
\red{A constant inflow we consider in this paper corresponds to the
incident wave with $\theta =0$ to the internal graph in \cite{FelHil2}, thus our
situation may be said to be somewhat simpler.
Instead we can apply the technique of spectral decomposition of the Grover walk
and thus obtain more detailed information of the stationary state, for
example, on not only the surface of the internal graph but the interior.
To generalize this theorem for some general $\theta$ is one of the
interesting future problems because 
this problem is deeply related to an extraction of structures of the internal graph 
by just observing the response to the input of the general incident wave i.e., the scattering, 
as is considered by, for examples, \cite{FelHil1,FelHil2} and \cite{FaGoGu,FaGu} for the discrete-time and continuous-time cases, respectively. }

This paper is organized as follows. 
In section \ref{Sec:demonstration}, we give a demonstration for the cycle graph with three vertices case comparing 
with the theoretical and numerical results. \red{This simple but fruitful example shall show what we intend.}
In section \ref{Sec:uniquely_existence}, we give the proof of Theorem~\ref{thm:existence}. 
The time evolution restricted to the internal graph $G_0$ is regarded as a dynamical mapping with the every time external injection.
This map is no longer a normal operator in general. 
So we consider a general eigen-problem by taking the Jordan decomposition 
and show that the system is always included in the stable generalized eigenspace~\cite{R}, 
which implies the convergence of this dynamical system to a fix point. 
In this section, we also characterise the center generalized eigenspace for the Grover walk 
which must be eliminated from the general solution of the linear equation for the stationary state. 
We convert the eigen-problem, whose computational basis are generated by arcs, 
to a vertex based operator's one; a non-linear eigenequation with respect to 
the Dirichlet random walk associated with the boundary $\delta V=\{u_1,\dots,u_r\}$. 
Combining this with the arc-based analysis, we completely characterize the center generalized eigenspace. 
The center generalized eigenspace is generated by the set of fundamental cycles of $G_0$ and 
the eigenvectors of the Dirichlet random walk whose supports have no-overlap to any vertices connected to the tails. 
In section \ref{Sec:perfect transmittion}, we give the proof of Theorem~\ref{thm:perfect_transmittion}.
We show that the perfect transmission can be derived from 
the Perron-Frobenius theorem with some combinatorial flow analysis 
and the unitarity of the time evolution operator on the whole system and the stationarity of $\psi_\infty$.
In section \ref{Sec:scattering}, we take the extension of 
the number of tails from $2$ to $r\geq 2$
and we show Theorem~\ref{thm:scattering}. 
Combining it with the characterization of the center generalized eigenspace in the previous section, 
we can conclude that the stationary state is expressed by the combinatorial flow which is orthogonal to the eigenspace
induced by the fundamental cycles of $G_0$. This means that an electric current on the circuit described by the internal graph $G_0$ gives 
the stationary state driven by the Grover walk. See Corollaries \ref{1} and \ref{2}. 
We further discuss on that relation between quantum walks and the electric circuit in the forthcoming paper~\cite{HSforth}. 
Finally we give the summary in Section 6.
\section{A demonstration}\label{Sec:demonstration}
In this section, to demonstrate our result, we treat the cycle with three vertices denoted by $C_3$ as the internal graph $G_0$. 
We set the initial state by $z=1$. 
Let the complex valued amplitude at each arc on the $n$-th iteration denoted as in Fig.~\ref{fig:one}. Then we have 
\begin{align*} 
r_{n+1} &= \frac{2}{3}(a_n+b_n)-\frac{1}{3},\\
t_{n+1} &= \frac{2}{3}(e_n+f_n), \\
\begin{bmatrix}a_{n+1} \\ b_{n+1} \\ c_{n+1} \\ d_{n+1} \\ e_{n+1} \\ f_{n+1}\end{bmatrix} &= 
\begin{bmatrix}
0 & 0 & 0 & 1 & 0 & 0\\
0 & 0 & 0    & 0   & -1/3 & 2/3 \\
-1/3 & 2/3 & 0 & 0 & 0 & 0\\
0 & 0 & 0 & 0 & 2/3 & -1/3 \\
2/3 & -1/3 & 0 & 0 & 0 & 0 \\
0 & 0 & 1 & 0 & 0 & 0
\end{bmatrix}
\begin{bmatrix}a_{n} \\ b_{n} \\ c_{n} \\ d_{n} \\ e_{n} \\ f_{n}\end{bmatrix}
+\begin{bmatrix}0 \\ 0 \\ 2/3 \\ 0 \\ 2/3 \\ 0\end{bmatrix}.
\end{align*}
Putting $\psi_n:={}^T[a_n,b_n,\dots,f_n]$ and denote the matrix in RHS by $E_{PON}$ and the second term vector corresponding to 
\red{inflow injecting constantly into the internal graph}
in RHS by $\rho$, then we rewrite
	\begin{align}
        \psi_0 &= 0; \notag \\
        \psi_{n+1} &= E_{PON}\psi_n+\rho.
        \end{align}
Remark that since the matrix $E_{PON}$ is a submatrix of the infinite dimensional unitary operator on the whole system, 
then $E_{PON}$ is no longer unitary. 
Furthermore, $E_{PON}$ is not a normal operator, which means that we can not take an orthogonal decomposition to $E_{PON}$. 
\red{The most important thing we expect is whether $\psi_n$ convergences or not as $n$ tends to $\infty$. }
To confirm it, we have done the numerical analysis on the eigenvalue of $E_{PON}$. See Fig.~\ref{fig:two} for the geometric expression for $\sigma(E_{PON})$ on the complex plain.
We can observe that all the absolute value of the eigenvalues except the $(+1)$-eigenvalue is strictly less then $1$. 
The eigenspace of the $(+1)$-eigenvalue is expressed by 
	\[ \mathbb{C}{}^T\begin{bmatrix} 1 & -1 & -1 & 1 & 1 & -1 \end{bmatrix}. \]
Using the induction with respect to time step $n$ and putting the above vector $\bs{\gamma}$, \red{that is,
$\bs{\gamma}={}^T\begin{bmatrix} 1 & -1 & -1 & 1 & 1 & -1 \end{bmatrix}$,}
we can state that $\psi_n$ is orthogonal to $\bs{\gamma}$. 
This orthogonality is still not enough to show that $\bs{\gamma}$ is in the complement invariant space of the invariant subspace including $\psi_n$'s 
since $E_{PON}$ is not a normal operator. 
However as we will see later 
the invariant space of $E_{PON}$, whose absolute value of the eigenvalue is $1$ denoted by $\mathcal{H}_c$, 
is orthogonal to its complement invariant subspace denoted by $\mathcal{H}_s$. 
Thus there are no contributions of such a $(+1)$-eigenspace to this time evolution, 
which implies the convergence \red{of} $\psi_n$ since the other absolute value of the eigenvalues are strictly smaller than $1$. 
Then we can solve the following inhomogeneous linear equation with confidence to obtain the stationary state $\psi_\infty$:
	\[ (1-E_{PON})\psi_\infty=\rho. \]
Remarking that since $1\in \sigma(E_{PON})$, $(1-E_{PON})$ is not invertible. 
The expression for the solution space so that the first term in the following has no overlap to $(+1)$-eigenspace is 
	\[ {}^T\begin{bmatrix}1/3 & 1/6 & 2/3 & 1/3 & 5/6 & 2/3\end{bmatrix}+\mathbb{C} \bs{\gamma}; \]
that is, $\psi_\infty ={}^T\begin{bmatrix}1/3 & 1/6 & 2/3 & 1/3 & 5/6 & 2/3\end{bmatrix}$.
Therefore the reflection and transmission rates are computed by $r_*=(2/3) (1/3+1/6)-1/3=0$ and $t_*=(2/3) (5/6+2/3)=1$, respectively. 
Moreover we can confirm that $\psi_\infty$ satisfies both (\ref{eq:stationary_vertex}) and (\ref{eq:stationary_edge}).  
\begin{figure}[htbp]
\begin{center}
	\includegraphics[width=80mm]{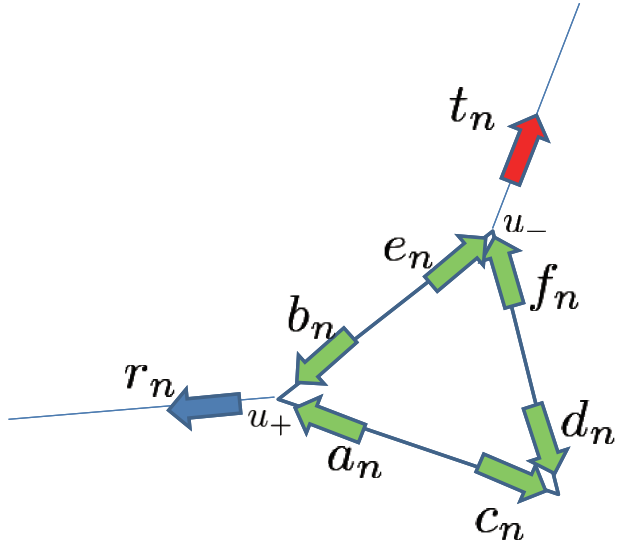}
\end{center}
\caption{The labeling of the complex amplitude at each arc of $C_3$ with tails at time $n$. }
\label{fig:one}
\end{figure}
\begin{figure}[htbp]
\begin{center}
	\includegraphics[width=80mm]{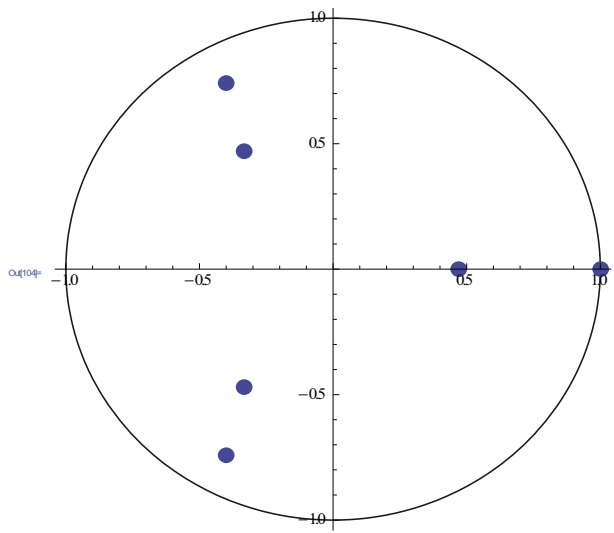}
\end{center}
\caption{The spectral distribution on the complex plain of $E_{PON}$ on $C_3$ with tails: 
the eigenvalue $(+1)$ comes from $\mathcal{L}^\perp$ which is an exceptional eigenspace defined in Sect.~3. }
\label{fig:two}
\end{figure}
%
\section{Proof of Theorem~\ref{thm:existence}: uniquely existence of the stationary state}\label{Sec:uniquely_existence}
\subsection{Proof of Theorem~\ref{thm:existence}}
Here we shall mainly give the proof of Theorem~\ref{thm:existance2} \red{stated soon later in this section}, 
which is a generalization of Theorem~\ref{thm:existence}; 
we can see the statement of Theorem~\ref{thm:existance2} coincides with that of Theorem~\ref{thm:existence}, 
if $r=2$, $\gamma_1=z=1$ and $\gamma_2=0$. 

Let $\{\mathbb{P}_s\}_{s=1}^{r}$ be the additional tails to give an infinite inflow toward the internal graph. 
The vertices of the tail $\mathbb{P}_s$ are labeled by $\{0_s, -1_s, -2_s,\cdots\}$, where $0_s\in V_0$. 
Define the sets of functions on $\tilde{A}$ and $\red{A_0}$ by $\mathbb{C}^{\tilde{A}}$ and $\mathbb{C}^{A_0}$, respectively.
Let $(\gamma_1,\dots,\gamma_{r})$ be a complex valued sequence assigned to each tails such that 
the initial state $\Psi_0\in \mathbb{C}^{\tilde{A}}$ is denoted by 
	\begin{equation}\label{eq:geneini} 
        \Psi_0(a)=\begin{cases} \gamma_s z^j & \text{: $o(a)\in V(\mathbb{P}_s)$, $t(a)=j_s$, $o(a)=(j-1)_s$, $(s=1,\dots,r)$, }\\ 0 & \text{: otherwise.} \end{cases} 
        \end{equation}
Here $z\in \mathbb{C}$ with $|z|=1$.         

%
We consider the following general setting as follows. 
\begin{assumption}\label{assumption}
\noindent 
\begin{enumerate}
\item $G_0$: connected and finite; 
\item {\it the unitarity of $U$} on $\mathbb{C}^{\tilde{A}}$ such that $U^*=U^{-1}$; 
\item {\it the free walk on the tails} such that 
$(U\Psi)(a_j^{(s)})=\Psi(a_{j+1}^{(s)})$ and $(U\Psi)(\bar{a}_j^{(s)})=\Psi(\bar{a}_{j-1}^{(s)})$ ($|j|\geq 1$, $1\leq s\leq r$). 
Here $a_j^{(s)}\in A(\mathbb{P}_s)$ with $|j|=\dist(G_o, t(a_j^{(s)}))>\dist(G_o, o(a_j^{(s)}))$; 
\item the initial state is (\ref{eq:geneini}). 
\end{enumerate}
\end{assumption}
%
\red{Let $\chi: \mathbb{C}^{\tilde{A}}\to \mathbb{C}^{A_0}$ be the boundary operator of $A_0$ such that for any $\Psi\in \mathbb{C}^{\tilde{A}}$, 
$(\chi\Psi)(a)=\Psi(a)$ $(a\in A_0)$. 
The adjoint $\chi^*: \mathbb{C}^{A_0}\to \mathbb{C}^{\tilde{A}}$ is described by 
	\[ (\chi^* \psi)(a) = \begin{cases} \psi(a) & \text{: $a\in A_0$,} \\ 0 & \text{: otherwise.} \end{cases} \]
Remark that $\chi\chi^*: \mathbb{C}^{A_0}\to \mathbb{C}^{A_0}$ is the identity operator of $\mathbb{C}^{A_0}$ 
and $\chi^*\chi: \mathbb{C}^{\tilde{A}}\to \mathbb{C}^{\tilde{A}}$ is the projection operator with respect to $A_0$. }
\begin{theorem}\label{thm:existance2}
We assume the above settings (1), (2), (3) and (4). 
Let $\Psi_n$ be the $n$-th iteration of the unitary evolution $U$ \red{for the initial state $\Psi_0$. Moreover} 
let us decompose $\Psi_n$ into $\Psi_n=\Psi_n^{(+)}\oplus \Psi_n^{(0)}\oplus \Psi_n^{(-)}$, where 
$\Psi_n^{(+)}:=z^{-n}\Psi_0$, $\Psi_n^{(0)}:=\chi^*\chi\Psi_n$ and $\Psi_n^{(-)}:=(1-\chi^*\chi)\Psi_n-\Psi_n^{(+)}$. 
Then 
	\[ \lim_{n\to\infty}z^{n}\Psi_n^{(\epsilon)}=\exists \varPhi_\infty^{(\epsilon)} \;\;(\epsilon\in\{0,\pm \}). \]
Moreover $\varPhi_\infty^{(\epsilon)}$'s satisfy
	\[ U(\varPhi_\infty^{(+)}+\varPhi_\infty^{(0)}+\varPhi_\infty^{(-)})=z^{-1}(\varPhi_\infty^{(+)}+\varPhi_\infty^{(0)}+\varPhi_\infty^{(-)}). \]
\end{theorem}
\noindent{\it Proof of Theorem~\ref{thm:existance2}.} \\
\quad Let $E_{PON}$ be the submatrix of the whole unitary time evolution operator $U$ restricted to $A_0$,  
that is, $E_{PON}=\chi U\chi^*$. 
Putting $\psi_n:=\chi \Psi_n$, we have 
	\begin{align*}
        \psi_n &= \chi \Psi_n=\chi U\Psi_{n-1}=\chi U \chi^*\chi \Psi_{n-1}+\chi U(1-\chi^*\chi) \Psi_{n-1}\\
               &= E_{PON}\psi_{n-1}+z^{-(n-1)}\chi U\Psi_0
        \end{align*}
for $n\geq 1$. 
Therefore in the internal graph, the dynamics is described by 
	\begin{equation}\label{eq:master_eq} 
        \psi_0=0,\;\; \psi_{n+1}=E_{PON}\psi_n+z^{-n}\rho, 
        \end{equation}
where $\rho$ is the ``external source" defined by $\rho:=\chi U\Psi_0$. 
Then (\ref{eq:master_eq}) implies 
	\[ \psi_n=z^{1-n}(1+z E_{PON}+\cdots+z^{n-1} E_{PON}^{n-1})\rho. \]
Thus our task will be \red{to analyze the spectra of $E_{PON}$ and to show the} convergence of $\psi_n$. 

Before going to spectral analysis on $E_{PON}$, we give the following remark. 
The time evolution operator $E_{PON}$ is no longer a normal operator, 
then $E_{PON}$ is not ensured the diagonalization in general. 
So from now on, we consider the Jordan decomposition. 
Recall a general theory on the linear algebra such that the generalized eigenspace of $\lambda$ is the invariant subspace spanned by 
$\{ \varphi \;|\; \exists m\geq 1 \mathrm{\;such\;that\;} (\lambda-E_{PON})^m\varphi=0 \}$. 
Moreover let us recall also that $E_{PON}$ is similar to 
	\[ \bigoplus_{\lambda_j \in \sigma(E_{PON})} J(\lambda_j,k_{j1})\oplus \cdots \oplus J(\lambda_j,k_{js(j)}). \]
Here $J(\lambda,k)$ is the $k$-dimensional matrix such that 
	\[J(\lambda,k)=
        \begin{bmatrix} 
        \lambda & 1       &         &        & \\
                & \lambda & 1       &        & \\
                &         & \ddots  & \ddots & \\
                &         &         & \ddots & 1 \\
                &         &         &        & \lambda
        \end{bmatrix}. \]
Since $\psi_n=z^{1-n}(1+z E_{PON}+\cdots+z^{n-1} E_{PON}^{n-1})\rho$, and our interest is the convergence in the limit of $n$, 
we remark that if $|\lambda|<1$, then 
	\begin{equation}\label{eq:Motomachi}
        \lim_{n\to\infty}\sum_{\ell=0}^{n}(zJ(\lambda,k))^\ell=z^{-1}
        \begin{bmatrix} 
        q & q^2 & q^3     & \cdots & q^k \\
          & q   & q^2     &        & q^{k-1} \\
          &     & \ddots  & \ddots & \vdots  \\
          &     &         & \ddots & q^2 \\
          &     &         &        & q  
        \end{bmatrix}, 
        \end{equation}
where $q:=z/(1-z\lambda)$ and all the \red{elements below the main diagonal} are zero, 
while if $|\lambda|\geq 1$, then the above LHS diverges. 
We introduce the invariant subspaces $\mathcal{H}_u$, $\mathcal{H}_{c}$ and $\mathcal{H}_{s}$ called 
a unstable generalized eigenspace, a center generalized eigenspace 
and a stable generalized eigenspace induced by $E_{PON}$, respectively~\cite{R}:
	\begin{align}
        \mathcal{H}_{u} &:= \spann\{ \psi\in\mathbb{C}^{A_0}  \;|\; \psi \mathrm{\;is\;a\;generalized\;eigenvector\;for\;an\;eigenvalue\;\lambda\;with\;|\lambda|>1} \}; \notag \\
        \mathcal{H}_{c} &:= \spann\{ \psi\in\mathbb{C}^{A_0}  \;|\; \psi \mathrm{\;is\;a\;generalized\;eigenvector\;for\;an\;eigenvalue\;\lambda\;with\;|\lambda|=1} \}; \notag \\
        \mathcal{H}_{s} &:= \spann\{ \psi\in\mathbb{C}^{A_0}  \;|\; \psi \mathrm{\;is\;a\;generalized\;eigenvector\;for\;an\;eigenvalue\;\lambda\;with\;|\lambda|<1} \}. \label{eq:generalized_ES}
        \end{align}
\red{Considering} the $n$-th iteration of our model, $\psi_n=z^{1-n}(1+z E_{PON}+\cdots+z^{n-1} E_{PON}^{n-1})\rho$, 
and (\ref{eq:Motomachi}), we \red{should concentrate on} to showing $\rho \in \mathcal{H}_s$ for the convergence of $z^{n}\psi_n$. 
\begin{lemma}\label{lem:H_u=0}
$\mathcal{H}_{u}=\emptyset$. 
\end{lemma}
\begin{proof}
Assume that $E_{PON}\varphi=\lambda \varphi$ holds. 
Then taking square modulus, we have 
$|\lambda|^2 ||\varphi||^2=||\chi U \chi^* \varphi||^2\leq ||U \chi^* \varphi||^2=||\chi^* \varphi||^2\leq ||\varphi||^2$. 
\end{proof}
The following statements are simple but the keys for the proof of the convergence. 
\begin{lemma}\label{lem:Koriyama}
If $|\lambda|=1$ and $\varphi_\lambda\in \ker(\lambda-E_{PON})$, then $E_{PON}^*\varphi_\lambda=\lambda^{-1}\varphi_\lambda$. 
Here $E_{PON}^*$ is the adjoint of $E_{PON}$.
\end{lemma}
\begin{proof}
Since $E_{PON}\varphi_\lambda=\lambda \varphi_\lambda$, $\chi U\chi^*\varphi_\lambda=\lambda \varphi_\lambda$ holds. 
In the following, let us see $\chi^*\chi U \chi^*\varphi_\lambda=U \chi^*\varphi_\lambda$ if $|\lambda|=1$. 
Assume that $\chi^*\chi U \chi^*\varphi_\lambda \neq U \chi^*\varphi_\lambda$. 
Then $\supp(U\chi^*\varphi_\lambda)\cap (A\setminus A_0)\neq \emptyset$, that is, there must exist $a\in A\setminus A_0$ such that $(U\chi^*\varphi_\lambda)(a)\neq 0$. 
Thus we have 
	\[||U\chi^*\varphi_\lambda||\gneq||\chi U\chi^*\varphi_\lambda||=||E_{PON}\varphi_\lambda||; \]
on the other hand, 
	\[ ||U\chi^*\varphi_\lambda||=||\chi^*\varphi_\lambda||=||\varphi_\lambda|| \]
since $U$ is unitary. This contradicts $E_{PON}\varphi_\lambda=\lambda\varphi_\lambda$ with $|\lambda|=1$. 
In result, 
	\begin{equation}\label{eq:hiroshima}
        U\chi^*\varphi_\lambda=\lambda\chi^*\varphi_\lambda
        \end{equation} 
if $E_{PON}\varphi_\lambda=\lambda\varphi_\lambda$ with $|\lambda|=1$. 
We can easily see $E_{PON}^*\varphi_\lambda=\lambda^* \varphi_\lambda$ remarking $|\lambda|=1$, $(\chi U\chi^*)^*=\chi U^* \chi^*$ and the unitarity of $U$. 
\end{proof}
By Lemma~\ref{lem:Koriyama}, we can show that for any $\lambda\in \sigma(E_{PON})$ with $|\lambda|=1$, 
	\[ \ker(\lambda-E_{PON})^m=\ker(\lambda-E_{PON}) \]
for any $m\geq 1$ as follows.
\begin{lemma}\label{eq:Fukushima}
$\mathcal{H}_c$ is an eigenspace; that is,  
	\begin{equation}
        \oplus_{|\lambda|=1}\ker(\lambda-E_{PON})=\mathcal{H}_c. 
	\end{equation}
\end{lemma}
\begin{proof}
For $\varphi\in \mathcal{H}_c$ with $(\lambda-E_{PON})^m \varphi=0$ $(m\geq 2)$ and 
$(\lambda-E_{PON})^{m-1} \varphi\neq 0,\dots, (\lambda-E_{PON}) \varphi\neq 0$,  
noting that $\phi:=(\lambda-E_{PON})^{m-1} \varphi\in \ker(\lambda-E_{PON})$, then 
	\[ 0\neq \langle (\lambda-E_{PON})^{m-1} \varphi,\phi  \rangle=\langle (\lambda-E_{PON})^{m-2} \varphi,(\lambda-E_{PON})^*\phi  \rangle \]
By Lemma~\ref{lem:Koriyama}, $(\lambda-E_{PON})^*\phi=0$. Thus $m$ must be $1$. 
\end{proof}
\begin{lemma}\label{lem:Shiroishi-Zao}
$\mathcal{H}_c \perp \mathcal{H}_{s}$. 
\end{lemma}
\begin{proof}
Let $\varphi_\lambda\in \ker(\lambda-E_{PON})\subset \mathcal{H}_c$ and $\varphi_{\mu}\in \ker(\mu-E_{PON})^m\subset \mathcal{H}_s$ with 
$(\mu-E_{PON})^m \varphi_\mu=0$ and $(\mu-E_{PON})^{m-1} \varphi_\mu\neq 0,\dots, (\mu-E_{PON}) \varphi_\mu\neq 0$. 
By (\ref{eq:Fukushima}), it is sufficient to check the orthogonality of $\varphi_\lambda$ and $\varphi_{\mu}$. 
Then 
	\begin{align} 
        0 &= \langle (\mu-E_{PON})^m\varphi_\mu, \varphi_\lambda \rangle 
           = \langle \varphi_\mu, {(\mu-E_{PON})^m}^*\varphi_\lambda \rangle \\
          &= {(\mu-\lambda)^m}^*\langle \varphi_\mu, \varphi_\lambda \rangle
        \end{align}
For the final equation, we used Lemma~\ref{lem:Koriyama}. Since $\mu \neq \lambda$, we obtain the orthogonality. 
\end{proof}
\begin{lemma}\label{lem:Sendai}
We have $\rho \in \mathcal{H}_s$. 
\end{lemma}
\begin{proof}
The center generalized eigenspace $\mathcal{H}_c$ is the eigenspace spanned by $\varphi_\lambda$'s. 
Here $\chi^*\varphi_\lambda\in \ker(\lambda-U)$ by (\ref{eq:hiroshima}). 
We examine the orthogonality. Lemma~\ref{lem:Koriyama} implies 
	\begin{align*}
         \langle \rho, \varphi_\lambda \rangle 
                 & = \langle \chi U\Psi_0, \varphi_\lambda \rangle 
                   = \lambda^{-1} \langle \Psi_0, \chi^*\varphi_\lambda \rangle 
                   = 0.
        \end{align*}
Then by Lemma~\ref{lem:Shiroishi-Zao}, we have obtained the desired conclusion. 
\end{proof}
By Lemma~\ref{lem:Sendai}, we have $E_{PON}^j \rho \in \mathcal{H}_s$ for any $j\geq 0$. 
Since $\psi_n=z^{1-n}(1+z E_{PON}+\cdots+z^{n-1} E_{PON}^{n-1})\rho$, then 
we obtain the existence $\lim_{n\to\infty}z^{n}\psi_n$. 

Finally, we consider a meaning of the existence of $\varphi_\infty:=\lim_{n\to\infty}z^{n}\psi_n$ to the whole system $\tilde{G}$.
Putting $\varphi_n:=z^n\psi_n$, we have $z^{-1}\varphi_{n+1}=E_{PON}\varphi_n+\rho$.
Then 
	\begin{equation}\label{eq:Yagiyama}
        z^{-1}\varphi_\infty=E_{PON}\varphi_\infty+\rho
        \end{equation}
holds. 
The state newly going outside of the internal graph at time $n$ is denoted by $t_n$; the support of $t_n$ is all the arcs 
of tails whose origins are $\delta V=\{u_1,\dots,u_r\}$. It holds that 
	\[ t_{n+1}=(1-\chi^*\chi)U\chi^*\psi_n+(z^{-n}(1-\chi^*\chi)U\Psi_0-z^{-(n+1)}\Psi_0). \]
\red{Here the first term of $t_n(\bar{e_j})$ corresponds to the value on $\mathbb{P}_j$ transmitted from the internal graph 
and the second one does to the value on $\mathbb{P}_j$ scattered at $t(e_j)$ from the external paths. }
Using (\ref{eq:Yagiyama}), we have 
	\[ t_{n+1}=z^{-(n+1)}(zU(\chi^*\varphi_n+\Psi_0)-(\chi^*\varphi_{n+1}+\Psi_0)). \]
Then putting $\tau_n:=z^nt_n$, we obtain the limit of $\tau_n$ such that
	\begin{equation}\label{eq:Aobayamaeki}
        \tau_\infty:=\lim_{n\to\infty}\tau_n=(zU-1)(\chi^*\varphi_\infty+\Psi_0). 
        \end{equation}
Now let us put $\Psi_{out}:=\tau_\infty+zU\tau_\infty+z^2U^2\tau_\infty+\cdots$, which represents the history of the out flow from the internal graph
for large time steps. 
Then we have 
	\begin{align*}
        \chi U(\chi^*\varphi_\infty+\Psi_0+\Psi_{out})
        	&= \chi U\chi^*\chi (\chi^*\varphi_\infty+\Psi_0+\Psi_{out})+ \chi U(1-\chi^*\chi)(\chi^*\varphi_\infty+\Psi_0+\Psi_{out}) \\
                &= E_{PON}\varphi_\infty+\rho \\
                &= z^{-1}\varphi_\infty 
        \end{align*}
by (\ref{eq:Yagiyama}). 
On the other hand, 
	\begin{align*}
        (1-\chi^*\chi) U(\chi^*\varphi_\infty+\Psi_0+\Psi_{out})
        	&= (1-\chi^*\chi) U(\chi^*\varphi_\infty+\Psi_0) +U\Psi_{out} \\
                &= z^{-1}(1-\chi^*\chi) (\tau_\infty+\varphi_\infty+\Psi_0) + z^{-1}\Psi_{out}-z^{-1}\tau_\infty \\
                &= z^{-1}(\Psi_0+\Psi_{out}). 
        \end{align*}
Here in the second equation, we used (\ref{eq:Aobayamaeki}), and $\Psi_{out}=\tau_\infty+zU\Psi_{out}$.  
Then putting $\Psi_\infty:=\chi^*\varphi_\infty+\Psi_0+\Psi_{out}$, we have  
	\begin{equation*}
        U\Psi_\infty=z^{-1}\Psi_\infty, 
        \end{equation*}
\red{which completes} the proof of Theorem~\ref{thm:existance2}. $\square$

\subsection{$\mathcal{H}_c$ for the Grover walk case}
In the previous subsection, we \red{observed} the convergence \red{in} the our model in the long time limit. 
\red{In order to} find the stationary solution, we \red{have to} solve the linear equation $\psi_\infty=E_{PON}\psi_{\infty}+\rho$, where 
$\psi_\infty:=\lim_{n\to\infty}z^{n-1}\psi_n$. 
Consider $z=1$ case. 
Then the solution is $\psi_\infty=(1-E_{PON}|_{\mathcal{H}_s})^{-1}\rho$ in formal. 
For more practical point of view, we need to eliminate the element of $\mathcal{H}_c$; this is the meaning of ``$E_{PON}|_{\mathcal{H}_s}$". 
\red{Therefore in this subsection, we confirm the consistency of the previous section and 
characterize $\mathcal{H}_c$ for the Grover walk case.} 

We introduce the incident matrix describing the incidence from an arc to its terminal vertex as follows: 
for any $u\in V_0$ and $a\in A_0$, 
	\[ (K)_{u,a}=\begin{cases} 1/\sqrt{\tilde{d}(u)} & \text{: $t(a)=u$;} \\ 0 & \text{: otherwise.} \end{cases} \]
\red{This} boundary operator $K$ satisfies the following properties. 
\begin{lemma}\label{lem:propertyK}
\noindent
\begin{enumerate}
\item $E_{PON}=S(2K^*K-I)$, where $(S\psi)(a)=\psi(\bar{a})$;
\item $KK^*=D$, where $(Df)(u)=(d(u)/\tilde{d}(u)) f(u)$;
\item $KSK^*=T$, where $T$ is the Dirichlet random walk operator on $G_0$ with the boundary $\delta V$; that is, 
	\[ (T)_{u,v}=\begin{cases}  1/\sqrt{\tilde{d}(u)\tilde{d}(v)} & \text{: $u$ and $v$ are adjacent in $G_0$,} \\ 0 & \text{: otherwise.} \end{cases} \]
\end{enumerate}
\end{lemma}
Therefore if we take the product of $K^*$ \red{and} $SK^*$ from the right to $E_{PON}$, 
we can see a relatively familiar self-adjoint operator, $T$, and the 
almost similar to the identity operator except \red{on} the boundaries, $D$, \red{respectively}. 
Indeed we have the following lemma using the above properties.
\begin{lemma}\label{lem:fromPONtoGON}
Let $L$ be the $2|V_0|\times |A_0|$ matrix such that $L=\begin{bmatrix} K^* & SK^* \end{bmatrix}$. 
Then we have 
	\[E_{PON}L=LE_{GON},\]
where 
	\[ E_{GON}=\begin{bmatrix} 0 & -I_{|V_0|} \\ 2D-I_{|V_0|} & 2T \end{bmatrix}. \]
\end{lemma}
%
We define $\mathcal{L}=\{K^*f+SK^*g \;|\;f,g\in \mathbb{C}^{|V_0|}  \}\subset \mathbb{C}^{|A_0|}$. 
Then Lemma~\ref{lem:fromPONtoGON} immediately implies $E_{PON}(\mathcal{L})\subset \mathcal{L}$. 
On the other hand, since $E_{GON}$ is an invertible $2|V_0|\times 2|V_0|$ matrix, 
then $E_{PON}(\mathcal{L})\supset \mathcal{L}$. 
Thus the following lemma holds. 
\begin{lemma}\label{lem:invariant}
The subspace $\mathcal{L}$ is invariant subspace under the action of $E_{PON}$; that is, 
	\[ E_{PON}(\mathcal{L})=\mathcal{L}.\]
\end{lemma}
When $E_{PON}$ would be a normal operator, $\mathcal{L}^\perp$ is the complement invariant subspace of $\mathcal{L}$ automatically. 
Now the normality of $E_{PON}$ does not hold. 
However without the normality of $E_{PON}$, the following statement still holds. 
\begin{lemma}\label{lem:invariantcomplement}
$E_{PON}$ can be decomposed into $E_{PON}=E_{PON}|_{\mathcal{L}}\oplus E_{PON}|_{\mathcal{L}^\perp}$, 
that is, 
the orthogonal complement subspace of $\mathcal{L}$ is the invariant subspace with respect to $E_{PON}$. 
More precisely, 
\[ E_{PON}(\mathcal{L}^\perp)=\mathcal{L}^\perp.  \]
\end{lemma}
\begin{proof}
The orthogonal complement $\mathcal{L}^\perp$ is expressed by 
	\begin{align} 
        \mathcal{L}^\perp
        & = \ker(K)\cap \ker(KS) \notag \\
        & = \left\{ \ker(S+1)\cap \ker(K)\right \} \oplus \left\{ \ker(S-1)\cap \ker(K)\right \}  \label{eq:machiko}
        \end{align}
It is easy to confirm that, for any $\psi_{\pm}\in \ker(S\pm 1)\cap \ker(K)$, $E_{PON}\psi_\pm=\pm\psi_\pm$ hold, respectively. 
Then $\mathcal{L}^\perp$ is an invariant subspace under the action of $E_{PON}$, that is, 
$E_{PON}(\mathcal{L}^\perp)=\mathcal{L}^\perp$. 
\end{proof}
Under the decomposition of $\mathcal{L}^\perp=\left\{ \ker(1-S)\cap \ker(K)\right \} \oplus \left\{ \ker(1+S)\cap \ker(K)\right \}$, 
we have $U_{\mathcal{L}^\perp}=-1 \oplus 1$ 
\red{and then} $\mathcal{L}^\perp\subset \mathcal{H}_c$. 
Since the external source $\rho$ is expressed by
	\begin{equation}\label{eq:extarnal_sourse} 
        \rho=\sum_{u_j\in \delta V} \frac{\gamma_{u_j} }{\sqrt{\tilde{d}(u_j)}}SK^*\delta_{u_j}, 
        \end{equation}
then we have $\rho\in \mathcal{L}$, 
\red{where for any $u\in V_0$, $\delta_u\in \mathbb{C}^{V_0}$ is the characteristic vector of $u$ such that
	\[ \delta_u(v)=\begin{cases} 1 & \text{: $u=v$,}\\ 0 & \text{: $u\neq v$.} \end{cases} \]}
Therefore by Lemma~\ref{lem:invariant}, we have $\psi_n\in \mathcal{L}$ for any $n\geq 0$ which is consistent with the previous section.
Remark that in the previous example of $G_0=C_3$, the eigenvalue $(+1)$ comes from this eigenspace $\mathcal{L}^\perp$. 

Now we can concentrate on the subspace $\mathcal{L}$. 
By Lemma~\ref{lem:fromPONtoGON}, the eigenequation $E_{PON}|_\mathcal{L}\psi=\lambda \psi$ for $\psi=K^*f+SK^*g$ is switched to 
	\[ L(\lambda-E_{GON})\begin{bmatrix}f \\ g \end{bmatrix}=0. \]
We give a useful characterization of $\ker L$. 
\begin{lemma}
	\[ \ker(L)=\ker(1-E_{GON}^2). \]
\end{lemma}
\begin{proof}
For any ${}^T[f\;g]\in \ker(L)$, it holds $K^*f+SK^*g=0$. 
Then taking product of $K$ and $KS$ from the \red{left}, we have 
	\[ Df+Tg=0,\;\;Tf+Dg=0,  \]
respectively. 
By the Gaussian elimination process, we have 
	\[ \ker(1-E_{GON}^2)=\ker\begin{bmatrix} D & T \\ T & D \end{bmatrix}. \]
Then we have ${}^T[f\;g]\in \ker(1-E_{GON}^{\red{2}})$. 
On the other hand, for any ${}^T[f\;g]\in \ker(1-E_{GON}^{\red{2}})$, it holds 
	\[ K(K^*f+SK^*g)=0,\;\;KS(K^*f+SK^*g)=0, \]
which is equivalent to $K^*f+SK^*g\in \mathcal{L}^\perp$. 
\red{With the fact} $K^*f+SK^*g\in \mathcal{L}$ by definition of $\mathcal{L}$, 
\red{we have} $K^*f+SK^*g=0$. Thus ${}^T[f\;g]\in \ker(L)$. 
\end{proof}
Therefore the eigenequation $(\lambda-E_{PON}|_\mathcal{L})\psi=0$ is equivalent 
to solving the following eigenequation
	\[ (1-E_{GON}^2)(\lambda-E_{GON})\phi=0, \;\;\phi\notin \ker(1-E_{GON}^2). \]

If $\lambda=\pm 1$, then $\phi\in \ker(1\mp E_{GON})^2\setminus \ker(1\mp E_{GON})$. 
However we will show in Lemma~\ref{lem:symmetricity} that this case can be also eliminated; that is, $\ker(1\mp E_{PON}|_{\mathcal{L}})=\{0\}$. 
Then $L\phi=0$.

If $\lambda\neq \pm 1$, then by Lemma~\ref{lem:fromPONtoGON}
	\begin{align}
        E_{PON}|_{\mathcal{L}}\psi &= \lambda \psi \notag\\
        & \Leftrightarrow (E_{GON}-\lambda) \begin{bmatrix} f_\lambda \\ g_\lambda \end{bmatrix}=0,\;\;\psi=K^*f_\lambda+K^*Sg_\lambda \notag \\
        & \Leftrightarrow f_\lambda=-\lambda^{-1}g_\lambda,\;g_\lambda\in\ker(\lambda^2-2\lambda T+(2D-1)). \label{eq:non-linear_eigen_eq}
        \end{align}
Thus if we could solve the eigenequation $\det(\lambda^2-2\lambda T+(2D-1))=0$ with respect to $\lambda$, 
we would obtain the spectrum of $E_{PON}$ directly, 
but it is hard to directly find an effective expression for the solution in our impression 
although we will use the expression (\ref{eq:non-linear_eigen_eq}) later. 

Then from now on, in the last half of this discussion, we take some consideration \red{on} $E_{PON}$ directly, 
without the consideration \red{on} $E_{GON}$, and finally we combine this consideration with (\ref{eq:non-linear_eigen_eq}), and address to show that 
all the non-negligible eigenstates of eigenvalues are $|\lambda|\lneq 1$. 
To this end, we consider the eigenequation 
\[ E_{PON}\varphi_\lambda=\lambda\varphi_\lambda \;(|\lambda|=1) \]
and find some properties of $\varphi_\lambda$. 

First, we give the following lemma. 
\begin{lemma}
Put $C':=2K^*K-I$. 
The operator $C'$ is decomposed into $C'=\oplus_{u\in V}C'_u$ under the space decomposition of 
	\[ \mathbb{C}^{|A_0|}=\bigoplus_{u\in V_0}\spann\{\delta_a \;|\; t(a)=u\} \]
For any $u\in V_0$, the local operator $C_u'$ is expressed by 
	\[ C_u'=\frac{2}{\tilde{d}(u)}J_{\tilde{d}(u)}-I_{\tilde{d}(u)}, \]
where $J_d$ and $I_d$ are the $d$-dimensional all $1$ matrix and identity matrix, respectively. 
Then we have 
	\begin{equation}
        \sigma(C_u')=\{2d(u)/\tilde{d}(u)-1, -1\}.
        \end{equation}
\end{lemma}
Since $2d(u)/\tilde{d}(u)-1<1$ for $u\in \delta V$, 
we have $||C_u'\psi||\leq ||\psi||$. 
Moreover since $E_{PON}=SC'$ and $S$ \red{are} unitary, \red{the relation} 
\[ E_{PON}\varphi_\lambda=\lambda\varphi_\lambda   \]
\red{implies that} $||\varphi_\lambda||\cdot|\lambda|= ||C'\varphi_\lambda||\leq ||\varphi_\lambda||$ which implies $|\lambda|<1$. 
This is consistent with Lemma~\ref{lem:H_u=0}.

Secondly, we show $\{\pm 1\}\notin \sigma(E_{PON}|_\mathcal{L})$; that is, the derivation of $(\pm 1)$-generalized eigenspace come from 
$\mathcal{L}^\perp\subset \mathcal{H}_c$. 
To this end, we \red{give} the following lemma related to the Kirchhoff condition on the boundary. 
\begin{lemma}\label{lem:Kirchhoff_boundary}
Let $\lambda\in \sigma(E_{PON})$ with $|\lambda|=1$ 
and we set $\varphi_{\lambda}\in \mathbb{C}^{|A_0|}$ by its eigenvector.
Then 
	\[ \sum_{a\in A_0 : t(a)=u_*} \varphi_\lambda(a)= \sum_{a\in A_0 : o(a)=u_*} \varphi_\lambda(a)=0\]
for every $u_*\in \{u_\pm\}$. 
\end{lemma}
\begin{proof}
The $(2d(u)/\tilde{d}(u))$-eigenstate and $(-1)$-eigenstate of $C_u'$ are 
	\begin{align}
        \ker[(2d(u)/\tilde{d}(u)-1)-C'_u] &= \mathbb{C}{}^T[1\;\;\;1], \\
        \ker[1+C'_u] &= \{ \varphi\in \mathbb{C}^{d(u)} \;|\; \sum_{a:t(a)=u}\varphi(a)=0 \}, 
        \end{align}
respectively for any $u\in V_0$. 
If $E_{PON}\varphi_\lambda=\lambda \varphi_\lambda$ with $|\lambda|=1$, then we have 
$C'\varphi_\lambda=\lambda S\varphi_\lambda$ since $S$ is a self-adjoint unitary.
Taking the norms \red{of} both sides, we have 
	\[ ||C'\varphi_\lambda||=||\varphi_\lambda|| \]
because $|\lambda|=1$. 
To conserve the norm, $\varphi_\lambda|_{t(a)\in u_*}$ ($u_*\in \delta V$) 
must belong to $(-1)$-eigenstate of $C_{u_*}$, 
since the absolute value of eigenvalue $2d(u_*)/\tilde{d}(u_*)-1$ for $u_*\in\delta V$ cannot be $1$ 
due to the setting of the graph with tails; $d(u_*)<\tilde{d}(u_*)$ for $u_*\in \delta V$. 
Therefore we have 
	\begin{equation} 
        \sum_{a:t(a)=u_*}\varphi_\lambda(a)=0. \label{eq:terminus}
        \end{equation}
On the other hand, $E_{PON}\varphi_\lambda=\lambda \varphi_\lambda$ implies 
	\begin{equation}\label{eq:balance} 
        \frac{1}{\tilde{d}(o(a))}\sum_{b:t(b)=o(a)}\varphi_\lambda(b)=\frac{\varphi_\lambda(\bar{a})+\lambda\varphi_\lambda(a)}{2}. 
        \end{equation}
for any $a\in A_0$ by definition of $E_{PON}$. 
So if $o(a)=u_*$, then by (\ref{eq:terminus}), the above equation is reduced to 
	\begin{equation*}
        \varphi_\lambda(\bar{a})=-\lambda\varphi_\lambda(a)
        \end{equation*}
which is equivalent to 
	\begin{equation}\label{eq:pseudo-symmeticity}
        \varphi_\lambda(a)=-\lambda\varphi_\lambda(\bar{a}),\;\;(t(a)=u_*).
        \end{equation}
Applying this to (\ref{eq:terminus}), we have 
	\[ \sum_{a:t(a)=u_*}\varphi_\lambda(a)=-\lambda \sum_{a:t(a)=u_*}\varphi_\lambda(\bar{a})=-\lambda \sum_{a:o(a)=u_*}\varphi_\lambda(a)=0. \]
Since $\lambda\neq 0$, we have 
	\[ \sum_{a:o(a)=u_*}\varphi_\lambda(a)=0. \]
\end{proof}
%
%
Using Lemma~\ref{lem:Kirchhoff_boundary}, we obtain the following lemma. 
\begin{lemma}\label{lem:symmetricity}
Let $\lambda\in \sigma(E_{PON}|_\mathcal{L})$ with $|\lambda|=1$. Then $\lambda\notin \{\pm 1\}$.
\end{lemma}
\begin{proof}
\red{Assume} $\lambda=1$. Let $E_{PON}|_{\mathcal{L}}\varphi_\lambda=\varphi_\lambda$ with $\varphi_\lambda\neq 0$; \red{we shall} show a contradiction. 
For \red{the case where} $\lambda =-1$, we can show it in a similar way. 
The eigenequation $E_{PON}\varphi_\lambda=\varphi_\lambda$ holds if and only if 
	\begin{align}\label{eq:DBC_seed}
        \frac{\varphi_\lambda(a)+\varphi_\lambda(\bar{a})}{2} &= \frac{1}{\tilde{d}(o(a))}\sum_{b:t(b)=o(a)}\varphi_\lambda(b)
         = \frac{1}{\tilde{d}(t(a))}\sum_{b:t(b)=t(a)}\varphi_\lambda(b) 
        \end{align}
The first equality comes from the definition of $E_{PON}$ and the second equality is obtained by changing the arc $a$ in the first equation 
to the inverse $\bar{a}$. Combining the second and third ones, we can notice that 
``$\sum_{t(a)=u}\varphi_\lambda$" is the value of the reversible measure on $u\in V_0$ of the isotropic random walk on the whole graph $\tilde{G}$.
Thus for any $u\in V_0$, there exists a constant $c$ which is independent of \red{the choice of} vertices such that for any $u\in V_0$, 
	\begin{equation}\label{eq:DBC} 
        \sum_{b:t(b)=u}\varphi_\lambda(b)=c\tilde{d}(u). 
        \end{equation}
Then for every $a\in A_0$, it holds that 
	\[ c= \frac{\varphi_\lambda(a)+\varphi_\lambda(\bar{a})}{2} = \frac{1}{\tilde{d}(o(a))}\sum_{b:t(b)=o(a)}\varphi_\lambda(b)
         = \frac{1}{\tilde{d}(t(a))}\sum_{b:t(b)=t(a)}\varphi_\lambda(b).  \]
By Lemma~\ref{lem:Kirchhoff_boundary}, choosing $a\in A_0$ such that $o(a)\in \delta V$, we have $c=0$, 
which implies 
	\begin{align}
        \varphi_\lambda(a)+\varphi_\lambda(\bar{a}) &= 0,\;(a\in A_0) \label{eq:ibuki}\\
        \sum_{b:t(b)=u}\varphi_\lambda(b) &= 0,\;(u\in V_0) \label{eq:akiko}.
        \end{align}
\red{Therefore (\ref{eq:ibuki}) and (\ref{eq:akiko}) imply $\varphi_\lambda\in\ker(S+1)$ and $\varphi_\lambda\in \ker K$, respectively. 
Then we have $\varphi_\lambda\in \mathcal{L}^\perp$ by (\ref{eq:machiko})}, which is the contradiction. 

\end{proof}

Thirdly, we show the generalized eigenspace of $E_{PON}$ with $|\lambda|=1$ are spanned by eigenvectors of $U$ restricted to $A_0$. 
To this end, now we combine the above statement obtained by arc based analysis with the fact (\ref{eq:non-linear_eigen_eq}) obtained by vertex based analysis. 
%
\begin{lemma}\label{lem:persistency}
Let $\lambda\in \sigma(E_{PON}|_\mathcal{L})$ with $|\lambda|=1$ and $g_\lambda$ be the function defined in (\ref{eq:non-linear_eigen_eq}). 
Then we have 
	\[ g_\lambda(u_*)=0 \]
for any $u_*\in \delta V$.
\end{lemma} 
\begin{proof}
By Lemma~\ref{lem:fromPONtoGON}, we can write 
	\begin{equation}\label{eq:expressionRW}
        \varphi_\lambda=K^*g_\lambda-\lambda SK^*g_\lambda. 
        \end{equation}
Using the Kirchhoff boundary condition on $u_*\in \delta V$ in Lemma~\ref{lem:Kirchhoff_boundary}, 
we have for the inflows to $u_*$, 
	\begin{align*} 
        \sum_{a:t(a)=u_*}\varphi_\lambda(a) 
         &= \sum_{a:t(a)=u_*}\left( \frac{g_\lambda(u_*)}{\sqrt{\tilde{d}(u_*)}}-\lambda \frac{g_\lambda(o(a))}{\sqrt{\tilde{d}(o(a))}}\right) \
         &= \frac{d(u_*)}{\sqrt{\tilde{d}(u_*)}}g_\lambda(u_*)-\lambda \eta(u_*)=0,
        \end{align*}
where $\eta(u_*) := \sum_{a:t(a)=u_*}g_\lambda(o(a))/\sqrt{\tilde{d}(o(a))}$. 
On the other hand, for the outflows from $u_*$, 
	\[ \sum_{a:o(a)=u_*}\varphi_\lambda(a)=-\lambda \frac{d(u_*)}{\sqrt{\tilde{d}(u_*)}}g_\lambda(u_*)+ \eta(u_*)=0. \]
The above two equations provide $(\lambda-\lambda^{-1})g(u_*)=0$. Since $\lambda\neq \pm 1$ by Lemma~\ref{lem:symmetricity}, 
we conclude that $g(u_*)=0$. 
\end{proof}

\red{Let $\Gamma$ be the set of all fundamental cycles in $G_0$; in particular, $\Gamma_{even}$ and $\Gamma_{odd}$ are the subsets of $\Gamma$ of even and odd length, respectively. 
Let $w_+:\Gamma\to \mathbb{C}^{A_0}$ be 
	\[ (w_+(\xi))(a_j)=1,\;\;(w_+(\xi))(\bar{a}_j)=-1 \]
for a fundamental cycle $\xi=(a_1,\dots,a_r)\in \Gamma$ with $t(a_1)=o(a_2),\dots,t(a_{r-1})=o(a_r)$ and $t(a_r)=o(a_1)$, otherwise $(w_+(\xi))(e)=0$. 
Moreover let $w_-:\Gamma_{even}\to \mathbb{C}^{A_0}$ be 
	\[ (w_-(\xi))(a_j)=(w_-(\xi))(\bar{a}_j)=(-1)^j \]
for an even length fundamental cycle $\xi=(a_1,\dots,a_r)\in \Gamma_{even}$; otherwise $(w_-(\xi))(e)=0$. 
It is easy to check that 
	\[ E_{PON}w_{\pm}(\xi)= \pm w_{\pm}(\xi). \]
Furthermore if $|\Gamma_{odd}|\geq 2$, then we fix an odd cycle $c_0$. For any other cycle $c$ of add length, we can find a closed path 
constructed by $c_0$ and $c$, say $c_0-c$. 
Here if $c_0$ and $c$ have some common vertex, then we find a closed path of even length $\tilde{c}=(\gamma_1,\dots,\gamma_{2n})$ such that 
$\gamma_i\neq \gamma_j,\bar{\gamma}_j$ if $i\neq j$. 
Thus we can define $w_-(c_0-c)$ as stated above. 
If $c_0$ and $c$ are disjoint, then there exists a path $p=(e_1,\dots,e_m)$ such that 
	\[ V(c_0)\cap V(P)=o(e_1),\;\; V(c)\cap V(P)=t(e_m). \]
We may set $o(e_1)=o(a_1)$ and $t(e_m)=o(b_1)$, where $c_0=(a_1,a_2,\dots,a_{2n-1})$ and $c=(b_1,b_2,\dots,b_{2\ell-1})$.
Then we define $w_-(c_0-c)$ as 
	\[ (w_-(c_0-c))(e)=\begin{cases} (-1)^i & \text{: $e=a_j$}\\ 2(-1)^{j-1} & \text{: $e=e_j$}\\ (-1)^{m-1+k}  & \text{: $e=b_k$}\\ 0 & \text{: otherwise}\end{cases} \]
and $(w_-(c_0-c))(\bar{e})=(w_-(c_0-c))(e)$. 
We can easily check that 
	\[ E_{PON}w_{-}(c_0-c)= - w_{-}(c_0-c). \]
}

We obtain the following theorem. 
\begin{theorem}\label{lem:kerEPON-lambda}
Let $\sigma_{per}:=\{ x\in \mathbb{R} \;|\; g\in \ker(x-T)\setminus\{0\},\; g(u)=0\;(u\in \delta V) \}\subset \sigma(T)$. 
Let us define the following three subspaces $\mathcal{C}_\pm$, $\mathcal{T}_{per} \subset \mathbb{C}^{A_0}$ 
such that  
\begin{align*}
\mathcal{C}_+ &:= \red{\mathrm{span}}\{ w_+(c) \;|\; c\in \Gamma \}; \\
\mathcal{C}_- &:= \red{\mathrm{span}}\{w_-(c) \;|\; c\in \Gamma_{even}\} \cup \{w_-(c_o-c) \;|\; c\in \Gamma_{odd}\setminus\{c_o\}\}; \\
\mathcal{T}_{per} &:= \bigoplus_{x\in \sigma_{per}}\{K^*g-e^{\pm i \arccos x} SK^*g \;|\; g\in \ker(x-T)\}.
\end{align*}
Then for the Grover walk case, 
the center generalized eigenspace $\mathcal{H}_c$ is expressed as follows: 
	\[ \mathcal{H}_c=\mathcal{C}_+ \oplus  \mathcal{C}_- \oplus \mathcal{T}_{per}. \]
Under this decomposition, the eigenvalues of the first and second terms are $\pm 1$ while those of the finial term are $\lambda$
with $|\lambda|=1$ and $\lambda \neq \pm 1$.

\end{theorem}
\begin{proof}
Let us consider the eigenequation $E_{PON}\varphi_\lambda=\lambda\varphi_\lambda$ with $|\lambda|=1$. 
By \red{Lemmas~\ref{lem:invariantcomplement} and ~\ref{lem:symmetricity}}, 
we have $\lambda=\pm 1$ if and only if $\varphi_\lambda\in \mathcal{L}^\perp\subset \mathcal{H}_c$. 
On the other hand, by (\ref{eq:non-linear_eigen_eq}), if $\lambda \neq \pm 1$, 
then $\varphi_\lambda=K^{*}g_\lambda+SK^{*}g_\lambda\in \mathcal{L}\cap \mathcal{H}_c$, 
where $g_\lambda\in \ker(\lambda^2-2\lambda T+(2D-1))$. 
\begin{enumerate}
\item \red{the case where} $\lambda \neq \pm 1$ case. 
From Lemma~\ref{lem:persistency}, 
for $\varphi_\lambda\in \mathcal{L}$ with $|\lambda|=1$, and for any $u\in V_0$, we have 
	\[ (Tg_\lambda)(u)=\zeta(\lambda)g_\lambda(u), \]
since $g_\lambda\in \ker(\lambda^2-2\lambda T+(2D-1))$ and $\mathrm{supp}(g_\lambda)\subseteq V_0\setminus \delta V$. 
Here $\zeta(\lambda)=(\lambda+\lambda^{-1})/2$. 
\item \red{the case where} $\lambda =\pm 1$. 
Recalling Lemma~\ref{lem:invariantcomplement}, 
we can state that the orthogonal complement space $\mathcal{L}^\perp$ is the invariant subspace with respect to $E_{PON}$, 
and under the decomposition of $\mathcal{L}^\perp=\ker K\cap \ker(1+S) \oplus \ker K \cap \ker(1-S)$, we have 
$E_{PON}|_{\mathcal{L}^\perp}=1\oplus -1$. 
By using the fact \cite{HKSS_JFA}, we have $\mathcal{C}_\pm=\ker K\cap \ker(1\pm S)$. 
\end{enumerate}

%
This completes the proof. 
\end{proof}

As a by-product of the above discussion, we obtain the following properties of (\ref{eq:non-linear_eigen_eq}). 
\begin{corollary}
Let $\lambda$ be a solution of 
	\[ \det(\lambda^2-2T\lambda+(2D-1))=0. \]
Then $\lambda$ satisfies the following properties:
\begin{enumerate}
\item $|\lambda|\leq 1$;
\item If $|\lambda|=1$, then we have 
\[ \ker(\lambda^2-2T\lambda+(2D-1))= \{ g\in \ker(\zeta(\lambda)-T) \;|\; g(u)=0\;\;\forall u\in \delta V\}. \] 
\end{enumerate}
\end{corollary}
%

\section{Proof of Theorem~\ref{thm:perfect_transmittion}: perfect transmission}\label{Sec:perfect transmittion}
\red{Let us recall that $d(u)$ and $\tilde{d}(u)$ denote the degrees of $u\in V_0$ in the internal graph $G_0$ and 
the whole infinite graph $\tilde{G}_0$, respectively. }
From this section, we restrict our consideration to the case where 
$r=2$, $\gamma_1=z=1$ and $\gamma_2=0$ in (\ref{eq:geneini}). 
We put $\Psi_\infty\in \mathbb{C}^{\tilde{A}}$ as the stationary state and $\psi_\infty:=\chi\Psi_\infty$. 
Let $\kappa_\infty:V_0\to \mathbb{C}$ such that 
	\[ \kappa_\infty(u)=\frac{2}{\sqrt{\tilde{d}(u)}}\langle K^*\delta_u,\psi_\infty \rangle. \]
Using this notation, the transmission and reflection rates are expressed by 
	\begin{align} 
        t_* &= \begin{cases} \kappa_\infty(u_-) & \text{: $u_+\neq u_-$,}\\ \kappa_\infty(u_-)+2/\tilde{d}_- & \text{: $u_+=u_-$,} \end{cases} \label{eq:transmittion}\\
        r_* &= \kappa_\infty(u_+)+\frac{2}{\tilde{d}_+}-1. \label{eq:reflection}
        \end{align}
Our target is to show $t_*=1$ \red{or} $r_*=0$.  
The stationary state restricted to the internal graph $\psi_\infty$ is described by $(1+E_{PON}+E_{PON}^2+\cdots)\rho$ and 
the external source $\rho$ is rewritten by (\ref{eq:extarnal_sourse}). 
Recall that the convergence is ensured by Theorem~\ref{thm:existence}. 
Then $\kappa_\infty(u)$ is reexpressed by
	\begin{align*}
        \kappa_\infty(u) &= \langle  K^*\delta_u, (1+E_{PON}+E_{PON}^2+\cdots)SK^*\delta_{u_+} \rangle \frac{2}{\sqrt{\tilde{d}}_+}\frac{2}{\sqrt{\tilde{d}(u)}}\\
        &= \langle  \delta_u, (\Xi_0+\Xi_1+\cdots)\delta_{u_+} \rangle \frac{2}{\sqrt{\tilde{d}}_+}\frac{2}{\sqrt{\tilde{d}(u)}},\\
        \end{align*}
where $\Xi_n:= KE_{PON}^nSK^*$. 
Then our interest is switched to the sequence of $\Xi_n$. 
We find the following three-term recursion relation of $\Xi_n$ as follows.
\begin{lemma}\label{lem:ChebyLike}
Let $\Xi_n$ be the above. Then we have 
	\begin{align*}
        \Xi_0 &= T; \\
        \Xi_{1} &= 2T^2-D; \\
        \Xi_{n} &= 2T\Xi_{n-1}-(2D-1)\Xi_{n-2}\;\;(n\geq 2).
        \end{align*}
\end{lemma}
We take the summation of $\Xi_n$ over $n$: $\xi_\infty:=\sum_{n=0}^\infty\Xi_n\delta_{u_+}$. 
Using Lemma~\ref{lem:ChebyLike}, we obtain
	\begin{equation*}
        (T-D)(2\xi_\infty+\delta_{u_+})=0.
        \end{equation*}
\red{Furthermore, } putting the $|V_0|$-dimensional degree matrix $M$ by $(Mf)(u)=\tilde{d}(u)f(u)$, 
\red{we can give another expression $\xi_\infty$ as}
	\[ \xi_\infty=\frac{\sqrt{\tilde{d}_+}}{4}M^{1/2}\kappa_\infty. \] 
Thus 
	\begin{equation}\label{eq:master_eq2}
        (T-D)\left(\frac{1}{2}M^{1/2}\kappa_\infty+\frac{1}{\sqrt{\tilde{d}_+}}\delta_{u_+}\right)=0 
        \end{equation}
holds. To solve (\ref{eq:master_eq2}), we need to clarify $\ker(T-D)$. 
\begin{lemma}\label{lem:PF}
Let $T$ and $D$ be the above. Then we have 
	\[ \ker(T-D)=\mathbb{C}\tilde{d}^{1/2}, \]
where $\tilde{d}^{1/2}(u):=\sqrt{\tilde{d}(u)}$. 
\end{lemma}
\begin{proof}
$T$ is expressed by $T=M^{-1/2}P'M^{1/2}$, where $P'$ is the transition Dirichlet random walk with the boundary $\delta V$.  
Then 
	\begin{align}
        (T-D)f=0 \Leftrightarrow (P'-D)M^{1/2}f=0 \Leftrightarrow (P_0-1)DM^{1/2}f=0, 
        \end{align}
where $P'D^{-1}=:P_0$ is the isotropic random walk on $G_0$ itself. 
Then by the Perron-Frobenius theorem, $DM^{1/2}f=c\tilde{d}$ for some constant $c$. Then we obtain the desired conclusion.
\end{proof}
Then we can reexpress LHS of (\ref{eq:master_eq2}) as follows using a constant $c$ such that 
	\begin{equation}\label{eq:star} 
        \frac{1}{2}M^{1/2}\kappa_\infty+\frac{1}{\sqrt{\tilde{d}_+}}\delta_{u_+}=c\tilde{d}^{1/2}. 
        \end{equation}
Inserting $u_+$ and $u_-$ into (\ref{eq:star}) and using the expressions of $t_*$ and $r_*$ in (\ref{eq:transmittion}), (\ref{eq:reflection}), 
respectively, we have 
	\begin{align}
        t_* &= \kappa_{\infty}(u_-)=2c; \label{eq:transmittion2}\\
        r_* &= \kappa_{\infty}(u_+)+2/\tilde{d}_+-1=2c-1. \label{eq:reflection2}
        \end{align}
\red{On the other hand}, from the unitarity of the time evolution of the whole system $U$ and the stationarity, 
we obtain the conservativeness with respective to the transmission and reflection rates. 
\begin{lemma}\label{lem:conservative}
Let $G'$ be a subgraph of the internal graph $G_o$ and $\delta A'_{in},\delta A'_{out}\subset \tilde{A}$ 
be the set of arcs connecting between $\tilde{G}$ and $G'$ such that for any $a\in A'_{in}$, $o(a)\in \tilde{V}\setminus V'$, 
$t(a)\in V'$ and $A'_{out}$ is the inverse of $A'_{in}$. 
Then we have 
	\[ ||\Psi_\infty|_{\delta A'_{in}} ||^2=||\Psi_\infty|_{\delta A'_{out}} ||^2. \]
In particular, let $t_*$ and $r_*$ be the above. Then we have 
	\[ t_*^2+r_*^2=1. \]
\end{lemma}        
\begin{proof}
Let $U$ be the unitary time evolution operator on the whole system $\tilde{A}$. 
Then by the stationarity of $\Psi_\infty$, we have 
	\[ U(\Psi_\infty|_{\delta A'_{in}}+\Psi_\infty|_{A'})=\Psi_{\infty}|_{\delta A'_{out}}+\Psi_\infty|_{A'}. \]
Taking the square norms \red{of} both sides, by the unitarity of $U$, we obtain
	$||\Psi_\infty|_{\delta A'_{in}} ||^2=||\Psi_\infty|_{\delta A'_{out}} ||^2. $
If we choose \red{the subgraph} $G'$ as \red{the internal graph} $G_0$, then $1=t_*^2+r_*^2$ holds. 
Then we obtained the desired conclusion. 
\end{proof}
Thus using the expressions (\ref{eq:transmittion2}) and (\ref{eq:reflection2}) \red{for $t_*$  and $r_*$}, respectively,  
we can solve the $c$'s satisfying the condition of Lemma~(\ref{lem:conservative}) by 
	\[ c=1/2\;\mathrm{or}\;c=0.  \]

If $c=1/2$, by (\ref{eq:transmittion2}) and (\ref{eq:reflection2}), 
\begin{equation}\label{eq:kappa_universarity} 
\kappa(u)=\begin{cases} t_*=1 & \text{: $u\notin u_+$,}\\ 1-2/\tilde{d}_+ & \text{: $u=u_+$} \end{cases}
\end{equation}
which implies $r_*=0$ and 
	\begin{equation}\label{eq:1} 
        \sum_{b\in \tilde{A}: t(b)=u}\Psi_\infty(b)=\tilde{d}(u)/2 \;\;(u\in V_0). 
        \end{equation}
by the definition of $K$. 
The following lemma is immediately obtained by 
$(U\Psi_\infty)(a)=\Psi_\infty(a)$ for any $a\in \tilde{A}$. 
%
\begin{lemma}\label{lem:average}
For every $a\in \tilde{A}$, 
	\[ \frac{\Psi_\infty(a)+\Psi_\infty(\bar{a})}{2}
        	= \frac{1}{\tilde{d}(o(a))}\sum_{b\in \tilde{A}:t(b)=o(a)}\Psi_\infty(b). \]
\end{lemma}
Then if $c=1/2$, then, inserting (\ref{eq:1}) into RHS in Lemma~\ref{lem:average}, we obtain 
	\[ \Psi_\infty(a)+\Psi_\infty(\bar{a})=1. \]
Now the final task to complete the proof of Theorem~\ref{thm:perfect_transmittion} is to eliminate the possibility of $c=0$.

If $c=0$, then by the same way as the $c=\sqrt{\tilde{d}_+}/2$ case, we have 
$t_*=0$, $r_*=-1$ and 
	\begin{align} 
        \Psi_\infty(a)+\Psi_\infty(\bar{a}) &= 0 \label{eq:flow1}\\
        \sum_{b\in \tilde{A}: t(b)=u}\Psi_\infty(b) &= 0 \label{eq:flow2}
        \end{align}
This is nothing but the combinatorial flow on the graph. 
Using the following combinatorial analysis completes our final task. 
We take a summation over all the flow in the internal arcs $A_0$. 
First we divide this summation into each vertex with respect to the inflows while the second one we divide it 
with respect to the outflows. 
It holds that 
	\begin{align*} 
        \sum_{u\in V_0}\sum_{a\in A_0:o(a)=u}\Psi_\infty(a)
                &=-\sum_{u\in V_0}\sum_{a\in A_0:t(a)=u}\Psi_\infty(a) \\
                &=-\left(\sum_{a\in A_0:t(a)=u_+}\Psi_\infty(a)+\sum_{a\in A_0:t(a)=u_-}\Psi_\infty(a)\right) \\
                &=1
        \end{align*}
Here we used (\ref{eq:flow1}) in the first equality and we used (\ref{eq:flow2}) in the second equality, the third equality 
comes from the external source. 
On the other hand, 
	\begin{align*}
        \sum_{u\in V_0}\sum_{a\in A_0:t(a)=u}\Psi_\infty(a)
        	&= \sum_{a\in A_0:t(a)=u_+}\Psi_\infty(a)+\sum_{a\in A_0:t(a)=u_-}\Psi_\infty(a) \\
                &= -1
        \end{align*}
\red{Both LHS's stated above coincide with $\sum_{a\in A_0}\Psi_\infty(a)$, which is a contradiction. 
Thus} we have reached to the desired conclusion. $\square$

\section{Proof of Theorem~\ref{thm:scattering}: scattering from a global view point}\label{Sec:scattering}
\noindent {\bf Proof of Theorem~\ref{thm:scattering}}: 
We notice that up to at least Lemma~\ref{lem:conservative}, there are no conflicts even if we extend the setting of 
the number of tails from $2$ to $r\geq  2$.
We insert the inflows $\{\alpha_1,\alpha_2,\dots, \alpha_{r}\}$ from each tail, that is, $z=1$, $\gamma_i=\alpha_i$ for $i=1,\dots,r$ in (\ref{eq:geneini}). 
Then just changing $\delta_{u_+}$ to 
$M^{-1/2}f_{in}$ in RHS of (\ref{eq:master_eq2}), we have 
\begin{equation}
(T-D)\left(\frac{1}{2}M^{1/2}\kappa_\infty+M^{-1/2}f_{in}\right)=0, 
\end{equation}
where $f_{in}: V_0\to \mathbb{C}$ such that
	\[ f_{in}(v)= \sum_{j:V(\mathbb{P}_j)\cap V_0=\{v\}}\alpha_j. \] 
Let $\delta V=\cup_{j=1}^r V(\mathbb{P}_j)\cap V_0. $
Then by Lemma~\ref{lem:PF}, we have 
	\begin{equation}\label{eq:newkappa}
        \frac{1}{2}M^{1/2}\kappa_\infty+M^{-1/2}f_{in}=\exists c\tilde{d}^{1/2}
        \end{equation}
Then using the expressions for the transmission and reflection amplitudes in (\ref{eq:transmittion2}), (\ref{eq:reflection2}), 
we obtain the outflow $\beta_j$ from the vertex $u\in V(\mathbb{P}_j)\cap V_0$ in the long time limit by 
	\begin{equation}\label{outflow}
        \beta_j=\frac{2}{\tilde{d}(u_j)}f_{in}(u)-\alpha_j+\kappa_\infty(u_j)=-\alpha_j+2c.
        \end{equation}
Since the time evolution operator is a real operator, applying Lemma~\ref{lem:conservative} and dividing the inflow and outflow 
into the real and imaginary parts; that is, 
$\sum_{j}(\mathrm{Re}\alpha_j)^2=\sum_{j}(\mathrm{Re}\beta_j)^2$ and 
$\sum_{j}(\mathrm{Im}\alpha_j)^2=\sum_{j}(\mathrm{Im}\beta_j)^2$, we have 
	\[ c=\mathrm{ave}(\alpha_1,\dots,\alpha_r) \mathrm{\;or\;} c=0. \]
\red{ Here $\mathrm{ave}(\alpha_1,\dots,\alpha_r):=(1/r)\sum_{j}\alpha_j$, which is the average of the inflows. }

\red{For} the same reason for the flow consistency discussed in the previous section, if $c=0$, then $c$ must be $c=\mathrm{ave}(\alpha_1,\dots,\alpha_r)=0$. 
Therefore inserting the new value $c$ into (\ref{outflow}), we obtain 
	\[ \beta_j=2\mathrm{ave}(\alpha_1,\dots,\alpha_r)-\alpha_j.  \] 
\red{Since} the outflow in the long time limit; $\beta_1,\dots,\beta_{r}$, \red{can be} are expressed by 
	\[ \begin{bmatrix}\beta_1 \\ \beta_2 \\ \vdots \\ \beta_{r}\end{bmatrix}
        = \begin{bmatrix} 
        2/r-1 & 2/r & \cdots & 2/r\\
        2/r & 2/r-1 & \cdots & 2/r\\
        \vdots       & \vdots         &  \ddots      & \vdots \\ 
        2/r & 2/r & \cdots & 2/r-1
         \end{bmatrix}
         \begin{bmatrix}\alpha_1 \\ \alpha_2 \\ \vdots \\ \alpha_{r}\end{bmatrix}, \]
we obtain (\ref{eq:GroverScat}) in Theorem~\ref{thm:scattering}.

Inserting the new value $c$ into (\ref{eq:newkappa}), we have 
	\begin{equation}\label{eq:shimidashi} 
        \frac{2}{\tilde{d}(u)}\sum_{a\in A_0:t(a)=u}\psi_\infty(a)=:\kappa_\infty(u)
        =\begin{cases} 
        2\mathrm{ave}(\alpha_1,\dots,\alpha_r) & \text{: $u\in V_0\setminus \{u_+\}$,} \\
        2\mathrm{ave}(\alpha_1,\dots,\alpha_r) -2\alpha_j/\tilde{d}_+ & \text{: $u=u_+$.}
         \end{cases} 
        \end{equation}
Then we obtain (\ref{eq:shinbashi}) in Theorem~\ref{thm:scattering}, 
and combining this with Lemma~\ref{lem:average}, we obtain (\ref{eq:shinbashi2}) in Theorem~\ref{thm:scattering}. 
We have completed the proof of Theorem~\ref{thm:scattering}. $\square$

\noindent\\
By Lemma~\ref{lem:average}, we also \red{obtain} the following corollary. 
\begin{corollary}\label{1}
        Let the setting be the same as in Theorem~\ref{thm:scattering}. 
        Then $\mathrm{j}(a):=\Psi_\infty(a)-\mathrm{ave}(\alpha_1,\dots,\alpha_r)$ is a combinatorial flow of $\tilde{G}$, that is;  
	\begin{align*}
        \sum_{b\in \tilde{A}\;:\; t(b)=u}\mathrm{j}(b)=0;\;\;\mathrm{j}(a)+\mathrm{j}(\bar{a})=0 
        \end{align*}
        for any $u\in \tilde{V}$ and $a\in \tilde{A}$. 
\end{corollary}
\begin{proof}
By (\ref{eq:shimidashi}), we have 
	\[ \frac{1}{|\tilde{d}(u)|}\sum_{b\in \tilde{A}:t(b)=u}\Psi_\infty(b)=\mathrm{ave}(\alpha_1,\dots,\alpha_r), \]
which implies $\sum_{b\in \tilde{A}\;:\; t(b)=u}\mathrm{j}(b)=0$ for any $u\in \tilde{V}$. 
By Lemma~\ref{lem:average}, 
	\[ \frac{\Psi_\infty(a)+\Psi(\bar{a})}{2}=\frac{1}{|\tilde{d}(u)|}\sum_{b\in \tilde{A}:t(b)=u}\Psi_\infty(b)=\mathrm{ave}(\alpha_1,\dots,\alpha_r) \]
for any $a\in \tilde{A}$. 
Then we can easily check that $\mathrm{j}(a)+\mathrm{j}(\bar{a})=0$. 
\end{proof}
Moreover we can state further property of $\mathrm{j}$ as follows. 
\begin{corollary}\label{2}
Let the setting be the same as in Theorem~\ref{thm:scattering} and let $\mathrm{j}$ be the \red{same as in Corollary~\ref{1}}. 
Then for any cycle $c=(a_1,\dots,a_s)$ with $t(a_1)=o(a_2),\dots,t(a_{s-1})=o(a_s),\; t(a_s)=o(a_1)$ in $G_0$, 
it holds 
	\[ \sum_{j=1}^s \mathrm{j}(a_j)=0. \]
\end{corollary}
\begin{proof}
By Theorem~\ref{lem:kerEPON-lambda}, $\psi_\infty$ must be orthogonal to $\mathcal{C}_+$; that is, $\langle w_+(c),\psi_\infty\rangle=0$ 
for any $c\in \Gamma$. 
Remark that if $\psi \perp \mathcal{C}_+$, then for any constant $\gamma\in \mathbb{C}$, 
$(\psi+\gamma) \perp \mathcal{C}_+$, where $(\psi+\gamma)(a):=\psi(a)+\gamma$, by the definition of $\mathcal{C}_+$. 
Then for any fundamental cycle $c=(a_1,\dots,a_r)\in \mathcal{C}_+$, 
	\[ 0=\langle w_+(c), \chi\mathrm{j}\rangle=\sum_{j=1}^s (\mathrm{j}(a)-\mathrm{j}(\bar{a}))=2\sum_{j=1}^s \mathrm{j}(a). \]
Here we used the skew symmetry of $\mathrm{j}$ in the final equation. 
Since for any cycle $c'$, $w_+(c')$ is expressed by a linear combination of $\{w_{+}(c)\}_{c\in \Gamma}$, we obtain the desired conclusion. 
\end{proof}
Corollaries~\ref{1} and \ref{2} correspond to Kirchhoff's current law for the constant resistance value $1$ 
and Kirchhoff's voltage law, respectively.
Then we conclude that the stationary state of our model driven by the Grover walk expresses an electric current on the circuit described by $G_0$. 
\\

\section{Summary}
\red{We considered the Grover walk on a graph with infinite length tails. 
The dynamics on these tails is the free quantum walk. 
We set the initial state so that the internal graph receives the same value inflows from the outside}
\blue{, that is,} 
\red{the tails, at every time step. 
Then we obtained the stationary state of this quantum walk and a kind of the Kirchhoff law of the stationary state} 
\blue{(Theorems~1.1, 1.2 and 3.2).}
\red{ 
Moreover the global scattering so that the internal graph can be regarded as a vertex, 
we obtain the local scattering manner of the Grover walk is reproduced in the long time limit}
\blue{(Theorem~1.3).}
\red{ 
Because of the stationarity of our quantum walk, we can expect to obtain a high relative probability at the marked vertex in a spatial quantum search 
without oscillation of the finding probability, e.g.,~\cite{Portugalbook} with respect to the time steps. 
Then the convergence speed to the stationary state is also one of the interesting future's problem. 
We further discuss on this relation between quantum walks and the electric circuit in the forthcoming paper~\cite{HSforth}. }

\noindent\\

\noindent {\bf Acknowledgments}
We thank to the referees for their useful comments. 
YuH's work was supported in part by Japan Society for the Promotion of Science Grant-in-Aid for Scientific Research 
(C) 25400208, (C)18K03401 and (A) 15H02055.
E.S. acknowledges financial supports from the Grant-in-Aid for 
Scientific Research (C) Japan Society for the Promotion of Science (Grant No.~19K03616) and Research Origin for Dressed Photon.
We are indebted to the ``Izakaya" restaurant, PONTA, in Sendai, Japan, for their hospitality, where we obtained the fundamental idea of this work. 



\begin{small}
\bibliographystyle{jplain}

\end{small}

\end{document}